\documentclass[envcountsame,runningheads,11pt]{llncs} 
\usepackage{times}

\usepackage{graphicx,amsmath,amssymb}
\usepackage{xspace}
\usepackage{url} 

\renewcommand{\circ}{\mathcal{C}}

\newcommand{\bbC}{{\mathbb{C}}}

\newcommand{\bbZ}{{\mathbb{Z}}}

\newcommand{\cO}{{\mathcal{O}}}

\newcommand{\cQ}{{\mathcal{Q}}}

\newcommand{\cU}{{\mathcal{U}}}

\newcommand{\tr}{\mathop{\mathrm{tr}}}

\newcommand{\poly}{\mathop{\mathrm{poly}}}
\newcommand{\E}{\mathop{\mbox{$\mathbf{E}$}}}

\def\l{\left}
\def\r{\right}

\renewcommand{\>}{\rangle}
\newcommand{\<}{\langle}
\newcommand{\ket}[1]{{|#1\rangle}}
\newcommand{\bra}[1]{\langle #1|}
\newcommand{\braket}[2]{\langle #1|#2\rangle}
\newcommand{\proj}[1]{\left|#1\right\>\!\left\<#1\right|}

\newcommand{\ot}{\otimes}
\newcommand{\eps}{\varepsilon}
\newcommand{\ra}{\to}
\DeclareMathOperator{\ad}{ad}
\DeclareMathOperator{\supp}{supp}
\DeclareMathOperator{\ANS}{ANS}

\newcommand{\be}{\begin{equation}}
\newcommand{\ee}{\end{equation}}
\def\ba#1\ea{\begin{align}#1\end{align}}
\def\bit{\begin{itemize}}
\def\eit{\end{itemize}}

\newcommand{\eq}[1]{(\ref{eq:#1})}
\newcommand{\thmref}[1]{Theorem~\ref{thm:#1}}

\newcommand{\lemref}[1]{Lemma~\ref{lem:#1}}
\newcommand{\secref}[1]{Section~\ref{sec:#1}}
\newcommand{\defref}[1]{Definition~\ref{def:#1}}
\newcommand{\fig}[1]{Figure~\ref{fig:#1}}

\begin{document}
\title{
  Superpolynomial speedups based on almost any quantum circuit
}

\titlerunning{Superpolynomial speedups based on almost any quantum
  circuit}

\author{Sean Hallgren\inst{1}
\and 
Aram W. Harrow\inst{2}}
\institute{
Department of Computer Science and Engineering,
The Pennsylvania State University \\
University Park, PA \\
\and
Department of Mathematics, University of Bristol,
Bristol, U.K.\\
\email{a.harrow@bris.ac.uk}}
\maketitle

\begin{abstract}
The first separation between quantum polynomial time and classical
bounded-error polynomial time was due to Bernstein and Vazirani in
1993.  They first showed a $O(1)$ vs. $\Omega(n)$ quantum-classical
oracle separation based on the quantum Hadamard transform, and then
showed how to amplify this into a $n^{O(1)}$ time quantum algorithm
and a $n^{\Omega(\log n)}$ classical query lower bound.

We generalize both aspects of this speedup.  We show that a wide
class of unitary circuits (which we call {\em dispersing} circuits)
can be used in place of Hadamards to obtain a $O(1)$ vs. $\Omega(n)$
separation.  The class of dispersing circuits includes all quantum
Fourier transforms (including over nonabelian groups) as well as
nearly all sufficiently long random circuits.  Second, we give a
general method for amplifying quantum-classical separations that
allows us to achieve a $n^{O(1)}$ vs. $n^{\Omega(\log n)}$ separation
from any dispersing circuit.
\end{abstract}

\section{Background}

Understanding the power of quantum computation relative to classical
computation is a fundamental question.  When we look at which problems
can be solved in quantum but not classical polynomial time, we get a
wide range: quantum simulation, factoring, approximating the Jones
polynomial, Pell's equation, estimating Gauss sums, period-finding,
group order-finding and even detecting some mildly non-abelian
symmetries~\cite{SICOMP::Shor1997,Hallgren2007,Watrous01,FriedlIMSS03,vDHI03}.  However, when we look at what algorithmic
tools exist on a quantum computer, the situation is not nearly as
diverse.  Apart from the BQP-complete problems~\cite{AJL06}, the main
tool for solving most of these problems is a quantum Fourier transform
(QFT) over some group.  Moreover, the successes have been for cases
where the group is abelian or close to abelian in some way.  For
sufficiently nonabelian groups, there has been no indication that the
transforms are useful even though they can be computed exponentially
faster than classically.  For example, while an efficient QFT for the
symmetric group has been intensively studied for over a decade because
of its connection to graph isomorphism, it is still unknown whether it
can be used to achieve any kind of speedup over classical
computation~\cite{STOC::Beals1997}.

The first separation between quantum computation and randomized
computation was the Recursive Fourier Sampling problem
(RFS)~\cite{SICOMP::BernsteinV1997}.  This algorithm had two
components, namely using a Fourier transform, and using recursion.
Shortly after this, Simon's algorithm and then Shor's algorithm for
factoring were discovered, and the techniques from these algorithms
have been the focus of most quantum algorithmic research
since~\cite{SICOMP::Simon1997:1474,SICOMP::Shor1997}.  These developed
into the hidden subgroup framework.  The hidden subgroup problem is an
oracle problem, but solving certain cases of it would result in
solutions for factoring, graph isomorphism, and certain shortest
lattice vector problems.  Indeed, it was hoped that an algorithm for
graph isomorphism could be found, but recent evidence suggests that
this approach may not lead to one~\cite{HMRRS06}.  As a way to
understand new techniques, this oracle problem has been very
important, and it is also one of the very few where super-polynomial
speedups have been found~\cite{IMSantha01,BaconCvD05}.

In comparison to factoring, the RFS problem has received much less
attention.  The problem is defined as a property of a tree with
labeled nodes and it was proven to be solvable with a quantum
algorithm 
super-polynomially faster than the best randomized algorithm.  This
tree was defined in terms of the Fourier coefficients over $\bbZ_2^n$.
The definition was rather technical, and it seemed that the simplicity
of the Fourier coefficients for this group was necessary for the
construction to work.  Even the variants introduced by
Aaronson~\cite{Aaronson2003} were still based on the same QFT over
$\bbZ_2^n$, which seemed to indicate that this particular abelian QFT
was a key part of the quantum advantage for
RFS.

The main result of this paper is to show that the RFS structure can be
generalized far 
more broadly.  In particular, we show that an RFS-style
super-polynomial speedup is achievable using almost any quantum
circuit, and more specifically, it is also true for any Fourier
transform (even nonabelian), not just over $\bbZ_2^n$.  This
illustrates a more general power that quantum computation has over
classical computation when using recursion.  The condition for a
quantum circuit to be useful for an RFS-style speedup is that the
circuit be {\em dispersing}, a concept we introduce to mean that it
takes many 
different inputs to fairly even superpositions over most of the
computational basis.

Our algorithm should be contrasted with the original RFS algorithm.
One of the main differences between classical and quantum computing is
so-called garbage that results from computing.  It is important in
certain cases, and crucial in recursion-based quantum algorithms
because of quantum superpositions, that intermediate computations are
uncomputed and that errors do not compound.  The original RFS
paper~\cite{SICOMP::BernsteinV1997} avoided the error issue by using
an oracle problem where every quantum state create from it had the
exact property necessary with no errors.  Their algorithm could have
tolerated polynomially small errors, but in this paper we relax this
significantly.  We show that even if we can only create states with
constant accuracy at each level of recursion, we can still carry
through a recursive algorithm which introduces new constant-sized
errors a polynomial number of times.

The main technical part of our paper shows that most quantum circuits
can be used to construct separations relative to appropriate oracles.
To understand the difficulty here, consider two problems that occur
when one tries to define an oracle whose output is related to the
amplitudes that result from running a circuit.  First, it is not clear
how to implement such an oracle since different amplitudes have
different magnitudes, and only phases can be changed easily.  Second,
we need an oracle where we can prove that a classical algorithm
requires many queries to solve the problem.  If the oracle outputs
many bits, this can be difficult or impossible to achieve.  For
example, the matrix entries of nonabelian groups can quickly reveal
which representation is being used.  To overcome these two problems we
show that there are binary-valued functions that can approximate the
complex-valued output of quantum circuits in a certain way.

One by-product of our algorithm is related to the Fourier transform of
the symmetric group.  Despite some initial promise for solving graph
isomorphism, the symmetric group QFT has still not found any
application in quantum algorithms.  One instance of our result is the
first example of a problem (albeit a rather artificial one) where the
QFT over the symmetric group is 
used to achieve a super-polynomial speedup.

\section{Statement of results}
Our main contributions are to generalize the RFS algorithm of
\cite{SICOMP::BernsteinV1997} in two stages.  First,
\cite{SICOMP::BernsteinV1997} described the problem of Fourier
sampling over $\bbZ_2^n$, which has an $O(1)$ vs. $\Omega(n)$ separation
between quantum and randomized complexities.  We show that here the
QFT over $\bbZ_2^n$ can be replaced with a QFT over any group, or for
that matter with almost any quantum circuit.  Next,
\cite{SICOMP::BernsteinV1997} turned Fourier sampling into recursive
Fourier sampling with a recursive technique.  We will generalize this
construction to cope with error and to amplify a larger class of
quantum speedups.  As a result, we can turn any of the linear speedups
we have found into superpolynomial speedups

Let us now explain each of these steps in more detail.  We replace the
$O(1)$ vs $\Omega(n)$ separation based on Fourier sampling with a
similar separation based on a more general problem called {\em oracle
  identification}.  In the oracle identification problem, we are given
access to an oracle $\cO_a:X\ra \{0,1\}$ where $a\in A$, for some sets
$A$ and $X$ with $\log |A|, \log |X|=\Theta(n)$.  Our goal is to determine the identity of $a$.  Further,
assume that we have access to a testing oracle $T_a:A\ra\{0,1\}$
defined by $T_a(a')=\delta_{a,a'}$, that will let us confirm that we
have the right answer.\footnote{This will later allow us to turn two-sided
into one-sided error; unfortunately it also means that a
non-deterministic Turing machine can find $a$ with a single query to
$T_a$.  Thus, while the oracle defined in BV is a candidate for
placing BQP outside PH, ours will not be able to place BQP outside of
NP.  This limitation appears not to be fundamental, but we will leave
the problem of circumventing it to future work.}

A quantum algorithm for identifying $a$ can be described as follows:
first prepare a state $\ket{\varphi_a}$ using $q$ queries to $\cO_a$,
then perform a POVM $\{\Pi_{a'}\}_{a'\in A}$ (with $\sum_{a'} \Pi_{a'}
\leq I$ to allow for the possibility of a ``failure'' outcome), using
no further queries to $\cO_a$.  The success probability is
$\bra{\varphi_a}\Pi_a \ket{\varphi_a}$.  For our purposes, it will
suffice to place a $\Omega(1)$ lower bound on this probability: say
that for each $a$, $\bra{\varphi_a}\Pi_a\ket{\varphi_a} \geq \delta$
for some constant $\delta>0$.   On the other hand, any classical algorithm
trivially requires $\geq\log (|A|\delta)=\Omega(n)$ oracle calls to identify
$a$ with success probability $\geq \delta$.  This is because each query
returns only one bit of information.
In \thmref{lin-sep} we will describe
how a large class of quantum circuits can achieve this $O(1)$
vs. $\Omega(n)$ separation, and in Theorems \ref{thm:QFT-dispersing} and
\ref{thm:random-dispersing} we will show specifically that QFTs and
most random circuits fall within this class.

Now we describe the amplification step.  This is a variant of the
\cite{SICOMP::BernsteinV1997} procedure in which making an oracle call
in the original problem requires solving a sub-problem from the same
family as the original problem.  Iterating this $\ell$ times turns query
complexity $q$ into $q^{\Theta(\ell)}$, so choosing $\ell=\Theta(\log
n)$ will yield the desired polynomial vs. super-polynomial separation.
We will generalize this construction by defining an amplified version
of oracle identification called {\em recursive oracle identification}.
This is described in the next section, where we will see how it gives
rise to superpolynomial speedups from a broad class of circuits.

We conclude that quantum speedups---even superpolynomial
speedups---are much more common than the conventional wisdom would
suggest.  Moreover, as useful as the QFT has been to quantum
algorithms, it is far from the only source of quantum algorithmic advantage.

\section{Recursive amplification}
\label{sec:recur}

In this section we show that once we are given a constant versus
linear separation (for quantum versus classical oracle
identification), we are able to amplify this to a super-polynomial
speedup.  We require a much looser definition than in
\cite{SICOMP::BernsteinV1997} because the constant case can have a
large error.

\begin{definition}\label{def:single-level}
For sets $A,X$, let $f: A\times X\to \{0,1\}$ be a function.  To set
the scale of the problem, let $|X|=2^n$ and $|A|=2^{\Omega(n)}$.
Define the set of oracles $\{\cO_a: a\in A\}$ by $\cO_a(x)=f(a,x)$,
and the states $\ket{\varphi_a}=\frac{1}{\sqrt{|X|}}\sum_{x\in
  X}(-1)^{f(a,x)}\ket{x}$.  The single-level oracle identification
problem is defined to be the task of determining $a$ given access to
$\cO_a$. 
Let $U$ be a family of quantum circuits,
implicitly depending on $n$.  We say that
$U$ solves the single-level oracle identification problem if
$$|\bra{a}U\ket{\varphi_a}|^2 \geq \Omega(1)$$
for all sufficiently large $n$ and all $a\in A$.  In this case, we
define the POVM $\{\Pi_a\}_{a\in A}$ by $\Pi_a = U^\dag\proj{a}U$.
\end{definition}

When this occurs, it means that $a$ can be identified from $\cO_a$
with $\Omega(1)$ success probability and using a single query.  In the
next section, we will show how a broad class of unitaries $U$ (the
so-called {\em dispersing} unitaries) allow us to construct $f$ for
which $U$ solves the single-level oracle identification problem.
There are natural generalizations to oracle identification problems
requiring many queries, but we will not explore them here.

\begin{theorem}\label{thm:superpoly-sep}
  Suppose we are given a single-level oracle problem with function $f$
  and unitary $U$ running in time $\poly(n)$.  Then we can construct a
  modified oracle problem from $f$ which can be solved by a quantum
  computer in polynomial time (and queries), but requires
  $n^{\Omega(\log n)}$ queries for any classical algorithm that
  succeeds with probability $\frac{1}{2}+n^{-o(\log n)}$.
\end{theorem}

We start by defining the modified version of the problem
(\defref{recursive} below), and
describing a quantum algorithm to solve it.  Then in
\thmref{recursive-correct} we will show that the quantum algorithm
solves the problem correctly in polynomial time, and in
\thmref{amplification}, we will show that randomized classical
algorithms require superpolynomial time to have a nonnegligible
probability of success.

The recursive version of the problem simply requires that another
instance of the problem be solved in order to access a value at a
child.  \fig{recursive} illustrates the structure of the
problem.

\vspace{-0.2cm}
\begin{figure}
\begin{center}
\includegraphics[width=.4\textwidth]{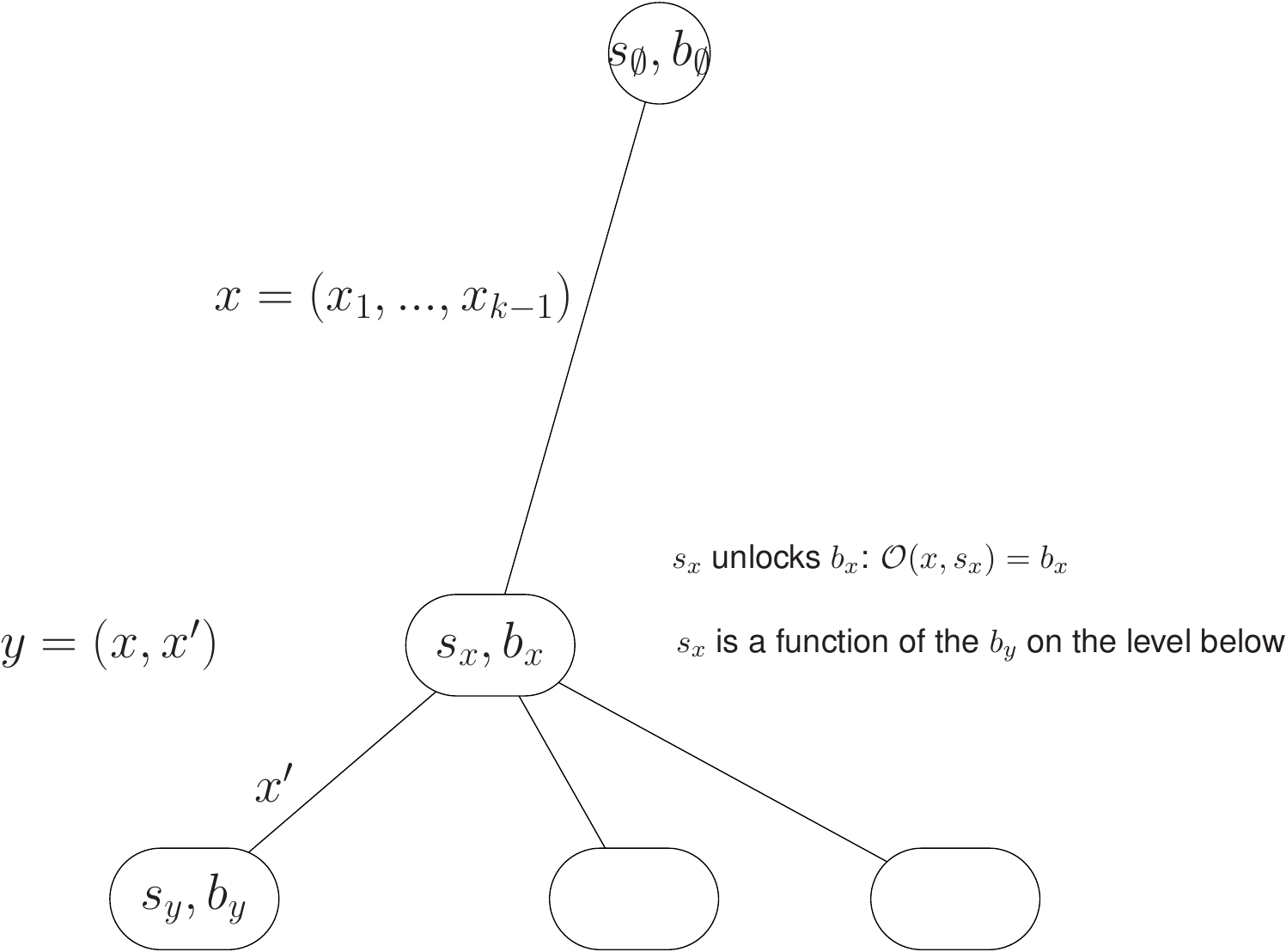}
  \caption{A depth $k$ node at location $x=(x_1,\ldots,x_k)$ is
    labeled by its secret $s_x$ and a bit $b_x$.  The secret $s_x$
    can be computed from the bits $b_y$ of its children, and once it
    is known,
    the bit $b_x$ is computed from the oracle $\cO(x,s_x)=b_x$.  If $x$
    is a leaf then it has no secret and we simply have $b_x=\cO(x)$.
 The goal is to compute the secret bit $b_\emptyset$ at
    the root. 
    \label{fig:recursive}}
\end{center}
\end{figure}
\vspace{-1cm}


Using the notation from \fig{recursive}, the relation between a secret
$s_x$, and the bits $b_y$ of its children is given by $b_y=
f(s_x,x')$, where $f$ is the function from the single-level oracle
identification problem.  Thus by computing enough of the bits
$b_{y_1}, b_{y_2},\ldots$
corresponding to children $y_1,y_2,\ldots$, we can solve the
single-level oracle identification problem to find $s_x$.  Of course
computing the $b_y$ will require finding the secret strings $s_y$,
which requires finding the bits of {\em their} children and so on, 
until we reach the bottom layer where queries return answer
 bits without the need to first produce secret strings.

\begin{definition}\label{def:recursive}
A level-$\ell$ recursive oracle identification problem is specified by
$X, A$ and $f$ from a single-level oracle identification problem
(\defref{single-level}), any function $s:\emptyset \cup X \cup X\times
X \cup \ldots \cup X^{\ell-1} \ra A$, and any final answer
$b_\emptyset\in\{0,1\}$.   Given these ingredients, an oracle $\cO$ is
defined which takes inputs in 
$$\bigcup_{k=0}^{\ell-1} \l[X^k \times A\r] \cup X^\ell$$ 
 and to return outputs in $\{0,1,\textsc{FAIL}\}$.  On inputs
 $x_1,\ldots,x_k\in X, a\in A$ with $1\leq k<\ell$, $\cO$ returns
\ba \cO(x_1,\ldots,x_k, a) &=
 f(s(x_1,\ldots,x_{k-1}), x_k) 
&& \text{when } a = s(x_1,\ldots,x_k)\\
\cO(x_1,\ldots,x_k, a) &=
\textsc{FAIL} && \text{when } a\neq s(x_1,\ldots,x_k).
\ea
If $k=0$, then $\cO(s(\emptyset))=b_\emptyset$ and
$\cO(a)=\textsc{FAIL}$ if $a\neq s(\emptyset)$.   When $k=\ell$,
$$\cO(x_1,\ldots,x_{\ell}) = f(s(x_1,\ldots,x_{\ell-1}), x_\ell).$$
The recursive oracle identification problem is to determine $b_\emptyset$
given access to $\cO$.
\end{definition}

Note that the function $s$ gives the values $s_x$ in \fig{recursive}.
These values are actually defined in the oracle and can be chosen
arbitrarily at each node.  Note also that the oracle defined here
effectively includes a testing oracle, which can determine whether $a
= s(x_1,\ldots,x_k)$ for any $a\in A, x_1,\ldots,x_k\in X$ with one
query.  (When $x=(x_1,\ldots,x_k)$, we use $s(x_1,\ldots,x_k)$ and
$s_x$ interchangeably.)  A significant difference between our
construction and that of \cite{SICOMP::BernsteinV1997} is that the
values of $s$ at different nodes can be set completely independently
in our construction, whereas \cite{SICOMP::BernsteinV1997} had a
complicated consistency requirement.

  {\bf The algorithm.} Now we turn to a quantum algorithm for the
  recursive oracle identification problem. If a quantum computer can
  identify $a$ with one-sided\footnote{One-sided error is a reasonable
    demand given our access to a testing oracle.  Most of these
    results go through with two-sided error as well, but for
    notational simplicity, we will not explore them here.} error
  $1-\delta$ using time $T$ and $q$ queries in the non-recursive
  problem, then we will show that the recursive version can be solved
  in time $O((q\frac{\log 1/\delta}{\delta})^\ell T)$.  For
  concreteness, suppose that
  $\ket{\varphi_a}=\frac{1}{\sqrt{|X|}}\sum_{x\in
    X}(-1)^{f(a,x)}\ket{x}$, so that $q=1$; the case when $q>1$ is an
  easy, but tedious, generalization.  Suppose that our identifying
  quantum circuit is $U$, so $a$ can be identified by applying the
  POVM $\{\Pi_{a'}\}_{a'\in A}$ with $\Pi_{a'}=U^\dag\proj{a'}U$ to
  the state $\ket{\varphi_a}$.

The intuitive idea behind our algorithm is as follows: At each level, we find
$s(x_1,\ldots, x_k)$ by recursively computing $s(x_1,\ldots,x_{k+1})$ for each
$x_{k+1}$ (in superposition) and using this information to create
many copies of $\ket{\varphi_{s(x_1,\ldots,x_k)}}$, from which we can
extract our answer.  However, we need to account for the errors
carefully so that they do not blow up as we iterate the recursion.  In
what follows, we will adopt the convention that Latin letters in kets
(e.g. $\ket{a}, \ket{x}, \ldots$) denote computational basis states,
while Greek letters (e.g. $\ket{\zeta},\ket{\varphi},\ldots$) are
general states that are possibly superpositions over many
computational basis states.  Also, we let the subscript $_{(k)}$
indicate a dependence on $(x_1,\ldots,x_k)$.  The recursive oracle
identification algorithm is as follows:

\begin{tabbing}
  ~~~~~ \= ~~~~ \= ~~~~ \= ~~~~ \= ~~~~ \= ~~~~ \= \kill
  {\bf Algorithm: \textsc{FIND}}\\
  {\bf Input:} $\ket{x_1,\ldots,x_k}\ket{0}$ for $k<\ell$\\
  {\bf Output:} $a_{(k)}=s(x_1,\ldots,x_k)$ up to error $\eps = (\delta/8)^2$,
  where $\delta$ is the constant from the oracle.
  This means\\
  $\ket{x_1,\ldots,x_k}\l[\sqrt{1-\eps_{(k)}}\ket{0}\ket{a_{(k)}}\ket{\zeta_{(k)}}
  + \sqrt{\eps_{(k)}}\ket{1}\ket{\zeta_{(k)}'}\r]$, where
  $\eps_{(k)}\leq \eps$ and $\ket{\zeta_{(k)}}$ and $\ket{\zeta_{(k)}'}$
  are arbitrary.\\
(We can assume this form without loss of generality by absorbing
phases into $\ket{\zeta_{(k)}}$ and $\ket{\zeta_{(k)}'}$.)
  \\
  {\bf 1.}\> Create the superposition
  $\frac{1}{\sqrt{|X|}}\sum_{x_{k+1}\in X}\ket{x_{k+1}}$.\\
  {\bf 2.}\> If $k+1<\ell$ then 
  let $a_{(k+1)}=\textsc{FIND}(x_1,\ldots,x_{k+1})$ (with error $\leq \eps$),
  otherwise $a_{(k+1)}=\emptyset$.\\
  {\bf 3.}\> Call the oracle $\cO(x_1,\ldots,x_{k+1},a_{(k+1)})$ to apply
  the phase $(-1)^{f(s(x_1,\ldots,x_k), x_{k+1})}$ using the key $a_{(k+1)}$.\\
  {\bf 4.}\> If $k+1<\ell$ then call \textsc{FIND}$^\dag$ to
  (approximately) uncompute 
  $a_{(k+1)}$.\\ 
  {\bf 5.}\> We are now left with $\ket{\tilde{\varphi}_{(k)}}$, which is close to
  $\ket{\varphi_{s(x_1,\ldots,x_k)}}$.\\
  \>Repeat steps 1--4 $m = \frac{4}{\delta}\ln\frac{8}{\delta}$ times to obtain
  $\ket{\tilde{\varphi}_{(k)}}^{\ot m}$\\
  {\bf 6.}\> Coherently measure $\{\Pi_a\}$ on each copy
  and test the
  results (i.e.\ apply $U$, test the result, and apply $U^\dag$).  \\
  {\bf 7.}\> If any tests pass, copy the correct $a_{(k)}$ to an output
  register, along with $\ket{0}$ to indicate success.\\
  \>  Otherwise put a $\ket{1}$ in the output to indicate failure.\\
  {\bf 8.}\> 
  Let everything else comprise the
  junk register $\ket{\zeta_{(k)}}$.\\
\end{tabbing}

\begin{theorem}\label{thm:recursive-correct}
  Calling $\textsc{FIND}$ on $|0\rangle$ solves the recursive oracle
  problem in quantum polynomial time.
\end{theorem}

\begin{proof}
The proof is by backward induction on $k$; we assume that the algorithm returns
with error $\leq \eps$ for $k+1$ and prove it for $k$.  The initial
step when $k=\ell$ is trivial since there is no need to compute
$a_{\ell +1}$, and thus no source of error.
If $k<\ell$, then assume that correctness of the algorithm has already
been proved for $k+1$.  Therefore Step 2 leaves the state
$$\frac{1}{\sqrt{|X|}}\sum_{x_{k+1}\in X} \ket{x_{k+1}}
\l[\sqrt{1-\eps_{(k+1)}}\ket{0}\ket{a_{(k+1)}}\ket{\zeta_{(k+1)}}
+ \sqrt{\eps_{(k+1)}}\ket{1}\ket{\zeta_{(k+1)}'}\r].$$
In Step 3, we assume for simplicity that the oracle was called
conditional on the success of Step 2.  This yields
$$\ket{\psi'_{(k)}}:=
\frac{1}{\sqrt{|X|}}\sum_{x_{k+1}\in X} \ket{x_{k+1}}
\l[(-1)^{f(a_{(k)}, x_{k+1})}
\sqrt{1-\eps_{(k+1)}}\ket{0}\ket{a_{(k+1)}}\ket{\zeta_{(k+1)}}
+ \sqrt{\eps_{(k+1)}}\ket{1}\ket{\zeta_{(k+1)}'}\r].$$
Now define the state $\ket{\psi_{(k)}}$ by
$$ \ket{\psi_{(k)}} :=
\frac{1}{\sqrt{|X|}}\sum_{x_{k+1}\in X} (-1)^{f(a_{(k)}, x_{k+1})}
\ket{x_{k+1}}
\l[\sqrt{1-\eps_{(k+1)}}\ket{0}\ket{a_{(k+1)}}\ket{\zeta_{(k+1)}}
+ \sqrt{\eps_{(k+1)}}\ket{1}\ket{\zeta_{(k+1)}'}\r].$$
Note that
$$\braket{\psi'_{(k)}}{\psi_{(k)}} = \frac{1}{|X|}
\sum_{x_{k+1}\in X} \left(1-\eps_{(k+1)} + (-1)^{f(a_{(k)},x_{k+1})}
  \eps_{(k+1)}\right).$$ This quantity is real and always $\geq
1-2\eps_{(k+1)}\geq\sqrt{1-4\eps}$ by the induction hypothesis.  Let
$$|\phi_{(k)}\rangle := \frac{1}{|X|} \sum_{x_{k+1}\in X}
(-1)^{f(a_{(k)},x_{k+1})} |x_{k+1}\rangle|0\rangle.$$ Note that
$\textsc{FIND}^\dag |x_1,\ldots,x_k,\psi_{(k)}\rangle =
|x_1,\ldots,x_k,\phi_{(k)}\rangle$. Thus there exists $\eps_{(k)}$ such
that applying $\textsc{FIND}^\dag$ to
$\ket{x_1,\ldots,x_k}\ket{\psi'_{(k)}}$ yields
$$\ket{x_1,\ldots,x_k}\otimes
\l[\sqrt{1-4\eps_{(k)}}\ket{\phi_{(k)}}+
\sqrt{4\eps_{(k)}}\ket{\phi'_{(k)}}\r],$$
where $\braket{\phi_{(k)}}{\phi_{(k)}'}=0$ and $\eps_{(k)}\leq
\eps$.  

We now want to analyze the effects of measuring $\{\Pi_{a}\}$ when we
are given the state $$\ket{\varphi_{(k)}} :=
\sqrt{1-4\eps_{(k)}}\ket{\phi_{(k)}}+ \sqrt{4\eps_{(k)}}\ket{\phi'_{(k)}}$$
instead of $\ket{\phi_{(k)}}$.  If we define $\|M\|_1=\tr\sqrt{M^\dag
  M}$ for a matrix $M$, then $\| \proj{\varphi_{(k)}}-\proj{\phi_{(k)}}\|_1
= 4 \sqrt{\eps_{(k)}}$ \cite{FG97}.  Thus
$$\bra{\varphi_{(k)}}\Pi_{a_{(k)}}\ket{\varphi_{(k)}} \geq
\bra{\phi_{(k)}}\Pi_{a_{(k)}}\ket{\phi_{(k)}} -
4\sqrt{\eps_{(k)}}
\geq \delta - 4\sqrt{\eps_{(k)}} \geq \delta/2.$$
In the last step we have chosen $\eps = (\delta/8)^2$.

Finally, we need to guarantee that with probability $\geq 1-\eps$ at
least one of the tests in Step 6 passes.  After applying $U$ and the
test 
oracle to $\ket{\varphi_{(k)}}$, we have $\geq \sqrt{\delta/2}$ overlap
with a successful test and $\leq \sqrt{1-\delta/2}$ overlap with an
unsuccessful test.  When we repeat this $m$ times, the amplitude in
the subspace corresponding to all tests failing is $\leq
(1-\delta/2)^{m/2}\leq e^{-m\delta/4}$.  If we choose
$m=(2/\delta)\ln(1/\eps)=(4/\delta)\ln(8/\delta)$ then the failure
amplitude will be $\leq \sqrt{\eps}$, as desired.

To analyze the time complexity, first note that the run-time is $O(T)$
times the number of queries made by the algorithm, and we have assumed
that $T$ is polynomial in $n$.  Suppose the algorithm at level $k$
requires $Q(k)$ queries.  Then steps 2 and 4 require $mQ(k+1)$ queries
each, steps 3 and 6 require $m$ queries each and together
$Q(k)=2mQ(k+1)+2m$.  The base case is $k=\ell$, for which $Q(\ell)=0$,
since there are no secret strings to calculate for the leaves.  The
total number of queries required for the algorithm is then
$Q(0)\approx (2m)^{2\ell}$.
If we choose $\ell = \log n$ the quantum query complexity will thus be
$n^{2\log 2m} = n^{O(1)}$ and the quantum complexity will be polynomial in $n$
compared with the $n^{\Omega(\log n)}$ lower bound.  
\end{proof}

This concludes the demonstration of the polynomial-time quantum
algorithm.  Now we turn to the classical $n^{\Omega(\log n)}$ lower
bound.  Our key technical result is the following lemma:

\begin{lemma}\label{lem:classical-lower-bound}
  Define the recursive oracle identification problem as above, with a
  function $f:A \times X\ra \{0,1\}$ and a secret $s:\emptyset \cup X
  \cup X\times X \cup \ldots \cup X^{\ell-1}\mapsto A$ encoded in an
  oracle $\cO$.  Fix a deterministic classical algorithm that makes
  $\leq Q$ queries to $\cO$.  Then if $s$ and $\ANS$ are chosen uniformly at
  random, the probability that $\ANS$ is output by the
  algorithm is
$$\leq \frac{1}{2} + \max\left( \frac{Q}{|A|^{1/3} - Q},
Q\left(\frac{\log|A|}{3}\right)^{-\ell}
\right).$$
\end{lemma}

Using Yao's minimax principle and plugging in 
$|A|=2^{\alpha n}$, $\ell=\log n$ and $Q=n^{o(\log n)}$ readily yields
\begin{theorem}\label{thm:amplification}
If $\log |A| = n^{\Omega(1)}$ and $\ell=\Omega(\log n)$, then any
randomized classical algorithm using  $Q = n^{o(\log n)}$ queries will
have $\frac{1}{2}+n^{-\Omega(\log n)}$ probability of successfully outputting
$\ANS$.
\end{theorem}

\begin{proof}[of \lemref{classical-lower-bound}]
 Let $T = \emptyset \cup X
  \cup \ldots \cup X^{\ell}$ denote the tree on which
  the oracle is defined.  We say that a node $x\in T$ has been {\em
    hit} by the algorithm if position $x$ has been queried by the
  oracle together with the correct secret, i.e.\ $\cO(s(x),x)$ has been
  queried.  
The only way to find to obtain information about $\ANS$ is
for the algorithm to query $\emptyset$ with the appropriate secret; in
other words, to hit $\emptyset$.

  For $x,y\in T$ we say that $x$ is an {\em ancestor} of $y$, and that
  $y$ is a {\em descendant} of $x$, if $y=x \times z $ for some $z\in
  T$.  If $z\in X$ then we say that $y$ is a {\em child} of $x$ and
  that $x$ is a {\em parent} of $y$.  Now define $S\subset T$ to be
  the set of all $x\in T$ such that $x$ has been hit but none of $x$'s
  ancestors have been.  Also define a function $d(x)$ to be the depth of a
  node $x$; i.e. for all $x\in X^k$, $d(x)=k$.  We combine these
  definitions to declare an invariant 
$$Z = \sum_{x\in S} \left(\frac{\log |A|}{3}\right)^{-d(x)}$$

The key properties of $Z$ we need are that:
\begin{enumerate}
\item Initially $Z=0$.
\item If the algorithm is successful then it terminates with $Z=1$.
\item Only oracle queries change the value of $Z$.
\item Querying a leaf can add at most $(\log
  |A|/3)^{-\ell}$ to $Z$.
\item Querying an internal node (i.e.\ not a leaf) can add at most $2
  / (|A|^{1/3} - Q)$ to  $\E Z$, where $\E$ indicates the
  expectation over random choices of $s$.
\end{enumerate}
Combining these facts yields the desired bound.

Properties 1--4 follow directly from the definition (with the
inequality in property 4 because it is possible to query a node that
has already been hit).  To establish property 5, suppose that the
algorithm queries node $x\in T$ and that it has previously hit $k$ of
$x$'s children.  This gives us some partial information about $s(x)$.
We can model this information as a partition of $A$ into $2^k$
disjoint sets $A_1,\ldots,A_{2^k}$ (of which some could be empty). From
the $k$ bits returned by the oracle on the $k$ children of $x$ we have
successfully queried, we know not only that $s(x)\in A$, but that
$s(x)\in A_i$ for some $i\in \{1,\ldots,2^k\}$.

We will now divide the analysis into two cases.  Either $k\leq
\frac{1}{3}\log|A|$ or $k>\frac{1}{3}\log|A|$.  We will argue that in
the former case, $|A_i|$ is likely to be large, and so we are unlikely
so successfully guess $s(x)$, while in the latter case even a successful
guess will not increase $Z$.  The latter case ($k>\frac{1}{3}\log|A|$)
is easier, so we consider it first.  In this case, $Z$ only changes if
$x$ is hit in this step and neither $x$ nor any of its ancestors have
been previously hit.  Then even though 
hitting $x$ will contribute $(\log |A|/3)^{-d(x)}$ to $Z$, it will
also remove the $k$ children from $S$ (as well as any other
descendants of $x$), which will decrease $Z$ by at least $k(\log
|A|/3)^{-d(x)-1} > (\log |A|/3)^{-d(x)}$, resulting in a net decrease
of $Z$.  

Now suppose that $k \leq \frac{1}{3}\log |A|$.  Recall that our
information about $s(x)$ can be expressed by the fact that $s(x)\in
A_i$ for some $i\in \{1,\ldots,2^k\}$.  Since the values of $s$ were chosen
uniformly at random, we have $\Pr(A_i)=|A_i|/|A|$.  Say that a set $A_i$
is {\em bad} if $|A_i| \leq |A|^{2/3}/2^k$.  Then for a particular bad
set $A_i$, $\Pr(A_i) \leq |A|^{-1/3}2^{-k}$.  From the union bound, we
see that the probability that {\em any} bad set is chosen is $\leq
|A|^{-1/3}$. 

Assume then that we have chosen a good set $A_i$, meaning that
conditioned on the values of the children there are $|A_i|\geq
|A|^{2/3}/2^k \geq |A|^{1/3}$ 
 possible values of $s(x)$.  However, previous failed queries at $x$
 may also have ruled out specific 
possible values of $x$.  There have been at most $Q$ queries at $x$,
so there are $\geq |A|^{1/3}-Q$ possible values of $s(x)$ remaining.
(Queries to any other nodes in the graph yield no information
on $s(x)$.)  Thus the probability of hitting $x$ is $\leq 1 /
(|A|^{1/3}-Q)$ if we have chosen a good set.  We also have a $\leq
|A|^{-1/3}$ probability of choosing a bad set, so the total probability
of hitting $x$ (in the $k \leq \frac{1}{3}\log |A|$ case) is $\leq
|A|^{-1/3} + 1/(|A|^{1/3}-Q) \leq 2 /
(|A|^{1/3}-Q)$.  Finally, hitting $x$ will increase $Z$ by at
most one, so the largest possible increase of $\E Z$ when querying a
non-leaf node is $\leq 2/(|A|^{1/3} - Q)$.  This completes the proof
of property 5 and thus the Lemma.
\end{proof}

\section{Dispersing Circuits}

In this section we define {\em dispersing} circuits and show how to
construct an oracle problem with a constant versus linear
separation from any such circuit.  In the next sections we will show
how to find dispersing circuits.  Our strategy for finding speedups
will be to start with a unitary circuit $U$ which acts on $n$ qubits
and has size polynomial in $n$.  We will then try to find an oracle
for which $U$ efficiently solves the corresponding oracle
identification problem.  
Next we need to define
a state $\ket{\varphi_a}$ that can be prepared with $O(1)$ oracle calls and
has $\Omega(1)$ overlap with $U^\dag\ket{a}$.  This is accomplished by
letting $\ket{\varphi_a}$ be a state of the form $2^{-n/2}\sum_x \pm \ket{x}$.
We can prepare $\ket{\varphi_a}$ with only two oracle calls (or one,
depending on the model), but to guarantee that $|\bra{a}U\ket{\varphi_a}|$
can be made large, we will need an additional condition on $U$.  For
any $a\in A$, $U^\dag\ket{a}$ should have amplitude that is mostly
spread out over the entire computational basis.  When this is the
case, we say that $U$ is {\em dispersing}.  The precise definition is
as follows:
\begin{definition}
Let $U$ be a quantum circuit on $n$ qubits.  
For $0<\alpha,\beta\leq 1$, we say that $U$ is
$(\alpha,\beta)$-dispersing if there exists a set $A\subseteq
\{0,1\}^n$ with $|A|\geq 2^{\alpha n}$ and
\be \sum_{x\in\{0,1\}^n}|\bra{a}U\ket{x}| \geq \beta
2^{\frac{n}{2}}. 
\label{eq:disp-def}\ee
for all $a\in A$.
\end{definition}
Note that the LHS of \eq{disp-def} can also be interpreted as the
$L_1$ norm of $U^\dag\ket{a}$.

The speedup in \cite{SICOMP::BernsteinV1997} uses $U=H^{\ot n}$, which
is (1,1)-dispersing since $\sum_x |\bra{a}H^{\ot n}\ket{x}|=2^{n/2}$ for
all $a$.  Similarly the QFT over the cyclic group is
(1,1)-dispersing.\footnote{Another possible way to generalize
  \cite{SICOMP::BernsteinV1997} is to consider other unitaries of the
  form $U=A^{\ot n}$, for $A\in\cU_2$.  However, it is not hard to
  show that the only way for such a $U$ to be $(\Omega(1),
  \Omega(1))$-dispersing is for $A$ to be of the form
  $e^{i\phi_1\sigma_z} H e^{i \phi_2\sigma_z}$.}
Nonabelian QFTs do not necessarily have the same strong dispersing
properties, but they satisfy a weaker definition that is still
sufficient for a quantum speedup.  Suppose that the measurement
operator is instead defined as $\Pi_a = U(\proj{a} \otimes I)U^\dag$, where
$a$ is a string on $m$ bits and $I$ denotes the identity operator on
$n-m$ bits. 
Then $U$ still permits oracle identification, but our
requirements that $U$ be dispersing are now relaxed.  Here, we give a
definition that is loose enough for our purposes, although further
weakening would still be possible.

\begin{definition}
Let $U$ be a quantum circuit on $n$ qubits.  
For $0<\alpha,\beta\leq 1$ and $0<m\leq n$, we say that $U$ is
$(\alpha,\beta)$-pseudo-dispersing if there exists a set $A\subseteq
\{0,1\}^m$ with $|A|\geq 2^{\alpha n}$ such that for all $a\in A$
there exists a unit vector $\ket{\psi}\in\bbC^{2^{n-m}}$ such that
\be \sum_{x\in\{0,1\}^n}|\bra{a}\bra{\psi}U\ket{x}| \geq \beta
2^{\frac{n}{2}}. 
\label{eq:ps-disp-def}\ee
\end{definition}
This is a weaker property than being dispersing, meaning that any
$(\alpha,\beta)$-dispersing circuit is also
$(\alpha,\beta)$-pseudo-dispersing.

We can now state our basic constant vs. linear query separation.
\begin{theorem}\label{thm:lin-sep}
If $U$ is $(\alpha,\beta)$-pseudo-dispersing, then there exists an oracle
problem which can be solved with one query, one use of $U$ and
success probability $(2\beta/\pi)^2$.   However, any classical
randomized algorithm that succeeds with probability $\geq \delta$
must use $\geq \alpha n + \log \delta$ queries.
\end{theorem}

Before we prove this Theorem, we state a Lemma about
how well states of the form $2^{-n/2}\sum_x e^{i\phi_x} \ket{x}$ can
be approximated by states of the form $2^{-n/2}\sum_x \pm
\ket{x}$.
\begin{lemma}\label{lem:complex-approx}
For any vector $(x_1,\ldots,x_d)\in\bbC^{d}$ there exists
$(\theta_1,\ldots,\theta_d)\in\{\pm 1\}^d$ such that 
$$\l|\sum_{k=1}^d x_k\theta_k\r|\geq \frac{2}{\pi}\sum_{k=1}^d \l|x_k\r|.$$
\end{lemma}
The proof is in the full version of the paper\cite{HH08}.

{\em Proof of \thmref{lin-sep}:}
Since $U$ is $(\alpha,\beta)$-pseudo-dispersing, there exists a set $A\subset
\{0,1\}^m$ with $|A|\geq 2^{\alpha n}$ and satisfying \eq{ps-disp-def}
for each $a\in A$.  The problem will be to determine $a$ by querying an
oracle $\cO_a(x)$.  No matter how we define the oracle, as long as it
returns only one bit per call any classical randomized algorithm
making $q$ queries can have success probability no greater than
$2^{q-\alpha n}$ (or else guessing could succeed with probability
$>2^{-\alpha n}$ without making any queries).  This implies the
classical lower bound.

Given $a\in A$, to define the oracle $\cO_a$, first use the definition
to choose a state $|\psi\rangle$ satisfying \eq{ps-disp-def}.  Then by
\lemref{complex-approx} (below), choose a vector $\vec{\theta}$ that (when
normalized to $\ket{\theta}$) will approximate the state
$U^\dag|a\rangle|\psi\rangle$. 
Define $\cO_a(x)$ so that $(-1)^{\cO_a(x)}=\theta_x=2^{n/2}\braket{x}{\theta}$.  By construction, \be
2^{-n/2}|\bra{a}\bra{\psi}U\ket{\theta}| \geq \frac{2}{\pi}\beta
\label{eq:pm-good-approx}\ee
which implies that creating $\ket{\theta}$, applying $U$, and measuring the
first register has probability $\geq (2\beta/\pi)^2$ of yielding the correct
answer $a$.
\qed

\section{Any quantum Fourier transform is pseudo-dispersing}

In this section we start with some special cases of dispersing
circuits by showing that any Fourier transform is dispersing.  In the
next section we show that most circuits are dispersing.

The original RFS paper~\cite{SICOMP::BernsteinV1997} used the fact that
$H^{\ot n}$ is (1,1)-dispersing to obtain their starting $O(1)$ vs
$\Omega(n)$ separation.  The QFT on the cyclic group (or any abelian
group, in fact) is also (1,1)-dispersing.  In fact, if we will accept
a pseudo-dispersing circuit, then any QFT will work:

\begin{theorem}\label{thm:QFT-dispersing}
Let $G$ be a group with irreps $\hat{G}$ and $d_\lambda$ denoting the
dimension of irrep $\lambda$.  Then the Fourier transform over
$G$ is $(\alpha,1/\sqrt{2})$-pseudo-dispersing, where
$\alpha=(\log\sum_\lambda d_\lambda)/\log|G| \geq 1/2$.
\end{theorem}
Via \thmref{lin-sep} and \thmref{superpoly-sep}, this implies that any
QFT can be used to obtain a superpolynomial quantum speedup.  For most
nonabelian QFTs, this is the first example of a problem which they can
solve more quickly than a classical computer.

\begin{proof}[Proof of \thmref{QFT-dispersing}]
Let $A=\{(\lambda,i): \lambda \in \hat{G}, i \in \{1,\ldots,d_\lambda\}\}$.

Let $V_\lambda$ denote the representation space corresponding to an
irrep $\lambda\in\hat{G}$.  The Fourier transform on $G$ maps vectors
in $\bbC[G]$ to superpositions of vectors of the form
$\ket{\lambda}\ket{v_1}\ket{v_2}$ for $\ket{v_1},\ket{v_2}\in
V_\lambda$.

Fix a particular choice of $\lambda$ and $\ket{i}\in V_\lambda$.
If $U$ denotes the QFT on $G$ then let
$$\rho = U^\dag\l(\proj{\lambda}\ot \proj{i} \ot 
\frac{I_{V_\lambda}}{d_\lambda}\r) U.$$ Define $V := \supp \rho$, and
let $\E_{\ket{\psi}\in V}$ denote an expectation over $\ket{\psi}$
chosen uniformly at random from unit vectors in $V$\footnote{We can
  think of $\ket{\psi}$ either as the result of applying a Haar
  uniform unitary to a fixed unit vector, or by choosing $\ket{\psi'}$
  from any rotationally invariant ensemble (e.g. choosing the real and
  imaginary part of each component to be an i.i.d. Gaussian with mean
  zero) and setting
  $\ket{\psi}=\ket{\psi'}/\sqrt{\braket{\psi'}{\psi'}}$.}  Finally,
let $\Pi$ be the projector onto $V$.  Note that $\rho =
\Pi/d_\lambda = \E \proj{\psi}$.

Because of the invariance of $\rho$ under right-multiplication by group
elements (i.e.\ $\bra{g_1}\rho\ket{g_2}=\bra{g_1h}\rho\ket{g_2h}$ for
all $g_1,g_2,h\in G$), 
we have for any $g$ that
\be  \bra{g}\rho\ket{g} = \frac{1}{|G|} \sum_h \bra{gh}\rho\ket{gh}
= \frac{1}{|G|} \tr(\rho) = \frac{1}{|G|}.
\label{eq:rho-isotropic}\ee
Since $\E\proj{\psi}=\rho$, \eq{rho-isotropic} implies that 
$$\E_{\ket{\psi}\in V}
|\braket{g}{\psi}|^2 = \bra{g}\rho\ket{g} = \frac{1}{|G|}.$$

Next, we would like to analyze $\E |\braket{g}{\psi}|^4$.
\ba \E_\ket{\psi} |\braket{g}{\psi}|^4 & = 
\E_\ket{\psi} \tr \l(\proj{g}\ot\proj{g}\r) \cdot
\l(\proj{\psi}\ot\proj{\psi}\r) \\
& = \tr \l(\proj{g}\ot\proj{g}\r)
\frac{I + \textsc{swap}}{d_\lambda(d_\lambda+1)}
\l(\Pi \ot \Pi\r)\label{eq:sym-subspace}\\
&\leq \tr \l(\proj{g}\ot\proj{g}\r)
\cdot (I + \textsc{swap}) (\rho \ot \rho)\\
& = 2(\bra{g}\rho\ket{g})^2 = \frac{2}{|G|^2}
\ea
To prove the equality on the second line, we use a
standard representation-theoretic trick (cf. section V.B of \cite{PSW05}).
First note that $\ket{\psi}^{\ot 2}$ belongs to the symmetric subspace
of $V\ot V$, which is a $\frac{d_\lambda(d_\lambda+1)}{2}$-dimensional
irrep of $\cU_{d_\lambda}$.  Since $\E_{\ket{\psi}}\proj{\psi}^{\ot 2}$
is invariant under conjugation by $u \ot u$ for any $u\in
\cU_{d_\lambda}$, it follows
that $\E_\ket{\psi} \proj{\psi}^{\ot 2}$ is proportional to a projector onto
the symmetric subspace of $V^{\ot 2}$.  Finally,
$\textsc{swap}\Pi^{\ot 2}$ has
eigenvalue $1$ on the symmetric subspace  of $V^{\ot 2}$ and
eigenvalue $-1$ on its orthogonal complement, the antisymmetric
subspace  of $V^{\ot 2}$.  Thus, $\frac{I+\textsc{swap}}{2}\Pi^{\ot
  2}$ projects onto 
the symmetric subspace and we conclude that 
$$\E_\ket{\psi} \proj{\psi}^{\ot 2} = 
\frac{(I+\textsc{swap})(\Pi \ot \Pi)}{d_\lambda(d_\lambda+1)}.$$


Now we note the inequality
\be \E |Y| \geq
(\E Y^2)^{\frac{3}{2}} / (\E Y^4)^{\frac{1}{2}},
\label{eq:fourth-moment}\ee
which holds for any random variable $Y$ and can be proved using
H\"{o}lder's inequality~\cite{SICOMP::Berger1997}.  Setting
$Y=|\braket{g}{\psi}|$, we can bound $\E_{\ket{\psi}}
|\braket{g}{\psi}| \geq 1/\sqrt{2|G|}$.  Summing over $G$, we find
$$\E_\ket{\psi} \sum_{g\in G}|\braket{g}{\psi}| \geq \frac{1}{\sqrt{2}}
\sqrt{|G|}.$$
Finally, because this last inequality holds in expectation, it must
also hold for at least some choice of $\ket{\psi}$.  Thus there exists
$\ket{\psi}\in V$ such that 
$$\sum_{g\in G}|\braket{g}{\psi}| \geq \frac{1}{\sqrt{2}}
\sqrt{|G|}.$$
Then $U$ satisfies the pseudo-dispersing condition in \eq{ps-disp-def}
for the state $\ket{\psi}$ with $\beta=1/\sqrt{2}$.

This construction works for each $\lambda\in\hat{G}$ and for $\ket{v_1}$
running over any choice of basis of $V_\lambda$.  Together, this
comprises $\sum_{\lambda\in\hat{G}} d_\lambda$ vectors in the set $A$.
\end{proof}

\section{Most circuits are dispersing}
\label{sec:random-circuits}

Our final, and most general, method of constructing dispersing
circuits is simply to choose a polynomial-size random circuit.  We
define a length-$t$ random circuit to consist of performing the
following steps $t$ times.
\begin{enumerate}
\item Choose two distinct qubits $i,j$ at random from $[n]$.
\item Choose a Haar-distributed random $U\in\cU_4$.
\item Apply $U$ to qubits $i$ and $j$.
\end{enumerate}
A similar model of random circuits was considered in \cite{DOP07}.
Our main result about these random circuits is the following Theorem.

\begin{theorem}\label{thm:random-dispersing}
For any $\alpha,\beta>0$, there exists a constant $C$ such that if
$U$ is a random circuit on $n$ qubits of length $t=Cn^3$ then $U$ is
$(\alpha,\beta)$-dispersing with probability
$$\geq 1 - \frac{2\beta^2 }{1-2^{-n(1-\alpha)}}.$$
\end{theorem}

\thmref{random-dispersing} is proved in the extended version of this paper\cite{HH08}.
The idea of the 
proof is to reduce the evolution of the fourth moments of the random
circuit (i.e. quantities of the form $\E_U \tr UM_1U^\dag
M_2UM_3U^\dag M_4$) to a classical Markov chain, using the approach of
\cite{DOP07}.  Then we show that this Markov chain has a gap of
$\Omega(1/n^2)$, so that circuits of length $O(n^3)$ have fourth
moments nearly identical to those of Haar-uniform unitaries from
$\cU_{2^n}$.  Finally, we use \eq{fourth-moment}, just as we did for
quantum Fourier transforms, to show that a large fraction of inputs
are likely to be 
mapped to states with large $L_1$-norm.  This will prove
\thmref{random-dispersing} and show that superpolynomial quantum
speedups can be built by plugging almost any circuit into the
recursive framework we describe in \secref{recur}.

\section*{Acknowledgments}
AWH was funded by the U.S. Army Research Office under grant
W9111NF-05-1-0294, the European Commission under Marie Curie grants
ASTQIT (FP6-022194) and QAP (IST-2005-15848), and the U.K. Engineering
and Physical Science Research Council through ``QIP IRC.''

\bibliographystyle{alpha}
\bibliography{rfs}

\appendix
\section{Most circuits are dispersing}
\label{app:dispersing}

In this Appendix we prove \thmref{random-dispersing}.

Suppose we start in a computational basis state $\ket{a}$, and after
$t$ steps of a random circuit (described in \secref{random-circuits}),
we have the state $\ket{\psi_t}$.  Let $U^\dag$ denote the 
circuit we have applied, and let $\psi_t$ denote
$\proj{\psi_t}$, so that $\psi_t = U^\dag \proj{a} U$.  For
$p\in\{0,1,2,3\}^n$ let $\sigma_p$ denote the 
tensor product of Pauli matrices $\sigma_{p_1}\otimes\cdots\otimes
\sigma_{p_n}$, where 
$\{\sigma_0, \sigma_1, \sigma_2, \sigma_3\}$ are the usual
single-qubit Pauli matrices $\{I,\sigma_z,\sigma_x,\sigma_y\}$.  Then we
can expand $\psi_t$ in the Pauli 
basis (following \cite{DOP07}) as
$$\psi_t = 2^{-\frac{n}{2}} \sum_p \gamma_t(p) \sigma_p,$$
where $\gamma_t(p) = 2^{-\frac{n}{2}} \tr \psi_t \sigma_p$. The
advantage of this approach is that each $\gamma_t(p)$ is real and
$\sum_p \gamma_t^2(p)=1$, so we can think of $\{\gamma_t^2(p)\}$ as a
probability distribution on $p\in\{0,1,2,3\}^n$.

Indeed, by an argument similar to that in \cite{DOP07}, we can show
how $\{\E_\circ \gamma_t^2(p)\}$ evolves in a 
way that can be described as a Markov chain on $\{0,1,2,3\}^n$.  (Here
the expectation is taken over the choice of random quantum circuit.) 

\begin{lemma}\label{lem:markov}
  Random quantum circuits are such that $\{\E_\circ \gamma_t^2(p)\}$
  evolve according to the following Markov chain on
  $\{0,1,2,3\}^n$:
  \bit
\item Select $i\neq j$ randomly from $[n]$.
\item If $p_i=p_j=0$ then do nothing.
\item Otherwise set $(p_i,p_j)$ to a random element of
  $\{0,1,2,3\}^2\backslash \{(0,0)\}$.
  \eit
Furthermore, this Markov chain satisfies the
  following properties.
\begin{enumerate}
 \item This Markov chain is irreducible and ergodic once we delete the
isolated vertex $0^n$.
\item Its stationary distribution (when starting with any physical
  state) has $\gamma^2(0^n)=2^{-n}$, but
  otherwise is uniform on $p\neq 0^n$,
  i.e.\ $\gamma^2(p)=4^{-n}/(1+2^{-n})$ for all $p\in\{0,1,2,3\}^n$.
\item Its spectral gap is $\geq \Omega(1/n^2)$.
\item There exists a constant $C$ such that if $t\geq Cn^3$ and if the
  initial state is a computational basis state then  $\E_\circ \gamma_t^2(p)
  \leq 4^{-n}$ for all $p\neq 0^n$.
\end{enumerate}
\end{lemma}

Before proving the Lemma, we show how it implies
\thmref{random-dispersing}.
Fix an
input $\ket{a}$, apply $t$ random two-qubit unitaries as described
above to yield the state $\ket{\psi_t}$, and define $Q_t := \sum_x
|\braket{x}{\psi_t}|^4$.  Expanding $\psi_t$ in 
terms of $\sigma_p$, we obtain
$$Q_t = 2^{-n}\sum_x \bra{x}\sum_{p\in\{0,1,2,3\}^n} \gamma_t(p)\sigma_p\ket{x}
\bra{x}\sum_{q\in\{0,1,2,3\}^n} \gamma_t(q)\sigma_q\ket{x}.$$
Now $\bra{x}\sigma_p\ket{x}$ will be zero if $p$ contains any 2's or
3's, since each of these lead to bit flips.  So we can restrict our
sum to $p$'s and $q$'s that are strings of 0's and 1's (corresponding
to $I$ and $\sigma_z$).  Moreover, if $p$ is such a string, then
$\bra{x}\sigma_p\ket{x} = (-1)^{p\cdot x}$.  Thus
$$Q_t = 2^{-n}\sum_{p\in\{0,1\}^n,q\in\{0,1\}^n}
\gamma_t(p)\gamma_t(q) \sum_x (-1)^{(p+q)\cdot x}
= \sum_{p\in\{0,1\}^n} \gamma_t^2(p).$$
To bound this sum, we use the last part of \lemref{markov} together
with $\gamma_t^2(0^n)=2^{-n}$ to find
$$\E_\circ Q_t \leq 2\cdot 2^{-n}.$$
Thus,
Markov's inequality implies that
$$\Pr_\circ\left(Q_t \geq \frac{2^{-n}}{\beta^2}\right) \leq 2\beta^2.$$

Now consider the event that $Q_t\leq 2^{-n}/\beta^2$.
We will
use this to lower-bound 
$\sum_x |\braket{x}{\psi_t}|$.   To do so, we define a random variable
$Y$ to equal $|\braket{x}{\psi_t}|$ for a uniformly
random choice of $x\in\{0,1\}^n$.  Then $\E_x Y^2 = 2^{-n}$, $\E_x Y^4
\leq \beta^{-2}\cdot 4^{-n}$, and by \eq{fourth-moment}, $\E_x Y = \E_x |Y| \geq
2^{-n/2}\beta$.  Thus $\sum_x |\braket{x}{\psi_t}| \geq \beta
2^{n/2}$.

Putting this together, we see that for any fixed input $\ket{a}$, and
for all but a $2\beta^2$-fraction of 
(sufficiently long) random circuits,
$$\sum_{x\in\{0,1\}^n}|\bra{x}U^\dag\ket{a}| \geq
\beta 2^{\frac{n}{2}}.$$
Say that a pair $(U,a)$ is {\em bad} when this does not hold.  So for
any $a$, the probability over $U\in \circ$ that $(U,a)$ is bad is $\leq 2\beta^2$.
Thus Markov's inequality implies that
$$\Pr_{U\in \circ}\left[\Pr_a [(U,a)\text{ is bad}] \geq 1-2^{-(1-\alpha)n}\right]
\leq \frac{2\beta^2}{1-2^{-(1-\alpha)n}}.$$
Turning this around, we conclude that a random circuit $U\in \circ$ is
$(\alpha,\beta)$-dispersing (meaning that $(U,a)$ is good for $\geq
2^{\alpha n}$ values of $a$) with probability $\geq 
1-2\beta^2/(1-2^{-(1-\alpha)n})$.   Thus, we can obtain a
quantum/classical separation 
from almost any quantum circuit with uniformly bounded parameters for
both the quantum upper bound and the classical lower bound.

Recent work\cite{HL08} has improved the analysis of the Markov chain to
show that the gap is $\Omega(1/n)$ and hence that circuits of length
$O(n^2)$ are generically dispersing.

{\em Proof of \lemref{markov}:} The reduction of the quantum random
circuit to a classical Markov chain is due to \cite{DOP07}, but we
will present an alternate, shorter proof in order to have a
self-contained presentation.

First, we show that $\gamma_t(p)^2$ follows a Markov chain.  Recall that
$\psi_t = 2^{-n/2}\sum_p \gamma_t(p) \sigma_p$ and let $W$ denote the
random unitary applied at time $t+1$.  Then $\psi_{t+1} = W\psi_t
W^\dag$.  Since $W$ acts only on two qubits, we can assume for the
purpose of this analysis that  $n=2$, so
$W\in\cU_4$.  Then
$$\gamma_{t+1}(p) = \frac{1}{2}\tr \sigma_p W \psi_t W^\dag
= \frac{1}{4}\sum_{q\in\{0,1,2,3\}^2}
 \gamma_t(q) \tr \sigma_p W \sigma_q W^\dag,$$
and
\ba \gamma_{t+1}^2(p) &= \frac{1}{16} \sum_{q,q'\in\{0,1,2,3\}^2}
 \gamma_t(q)\gamma_t(q')
\tr \sigma_p W \sigma_q W^\dag \tr \sigma_p W \sigma_{q'} W^\dag
\\ &= \sum_{q,q'\in\{0,1,2,3\}^2}\gamma_t(q)\gamma_t(q')
\bra{p,p}\ad_W^{\ot 2}\ket{q,q'}.\ea
Here $\ad_W$ is an operator on $\bbC^{16}$ defined by
$\bra{p}\ad_W\ket{q} := \tr\sigma_p W\sigma_q W^\dag / 4$ for
$p,q\in\{0,1,2,3\}^{2}$.
When we take the expectation over random choices of $W$, we obtain
\be \E \gamma_{t+1}^2(p) = 
\sum_{q,q'\in\{0,1,2,3\}^2}\gamma_t(q)\gamma_t(q')
\bra{p,p}\l(\E_W \ad_W^{\ot 2}\r)\ket{q,q'}. \ee
Define
\be\ket{\xi}=\frac{1}{\sqrt{15}}\sum_{p\in \{0,1,2,3\}^2\backslash\{(0,0)\}}
\ket{p}\ket{p} .
\label{eq:xi-def}\ee
We claim that $\E_W \ad_W^{\ot 2} = \proj{00} + \proj{\xi}$.  Since
$\braket{p,p}{\xi}\braket{\xi}{q,q'}$ is equal to $1/15$ when $p\neq
00$ and $q=q'\neq 00$, and zero otherwise, we will have
$$\E_W \bra{p,p}\ad_W^{\ot 2}\ket{q,q'} =
\left\{
\begin{array}{cc}
1 & \text{ if } p=q=q'=00\\
\frac{1}{15} & \text{ if } p\neq 00 \text{ and } q=q'\neq 00\\
0 & \text{ otherwise}
\end{array}
\right. ,$$
as claimed in the lemma.

To prove that $\E_W \ad_W^{\ot 2} = \proj{00} + \proj{\xi}$, we will
use representation theory. Schur's 
Lemma implies that $\E_W \ad_W^{\ot 2}$ is a
projector onto the invariant subspace of $\ad_W^{\ot 2}$.  For any 
integers $\lambda_1\geq \ldots \geq \lambda_d$, we have an
irrep of $\cU_d$, which we call $\cQ_\lambda^d$.
The conjugate irrep of $\cQ_\lambda^d$,
obtained by taking the complex conjugate of each representation
matrix, is given by $(\cQ_{\lambda}^d)^* \cong \cQ_{\lambda'}^d$, where
$\lambda' = (-\lambda_d,  - \lambda_{d-1}, \ldots,-\lambda_1)$.

The simplest non-trivial $\cU_d$ representation is  the defining
representation $\cQ^d_{(1)}\cong \bbC^d$, where $U\in\cU_d$ is mapped to
itself.    Here we
have dropped trailing zeros, so $(1)$ is equivalent to $(1,0,0,0)$.
Now observe
that $W\ra \ad_W$ is a representation of $\cU_4$ that is equivalent to
$\cQ^4_{(1)} \ot (\cQ^4_{(1)})^* \cong \cQ^4_{(1)} \ot \cQ^4_{(0,0,0,-1)}$.
Now we apply the
Clebsch-Gordan decomposition of the adjoint representation $\cQ^4_{(1)}
\ot \cQ^4_{(0,0,0,-1)}$ to obtain
$$ \cQ^4_{(1)}\ot \cQ^4_{(0,0,0,-1)}  \cong
\cQ^4_{(0)} \oplus \cQ^4_{(1,0,0,-1)},$$
meaning the direct sum of a trivial irrep $\cQ^4_{(0)}$ and a
15-dimensional irrep $\cQ^4_{(1,0,0,-1)}$. To
find the invariant subspace of $\ad_W^{\ot 2}$, we will need to
find the invariant subspace of
$$(\cQ^4_{(0)} \oplus \cQ^4_{(1,0,0,-1)})^{\ot 2} =
(\cQ^4_{(0)})^{\ot 2} \oplus
(\cQ^4_{(0)} \ot \cQ^4_{(1,0,0,-1)}) \oplus
(\cQ^4_{(1,0,0,-1)} \ot \cQ^4_{(0)}) \oplus
(\cQ^4_{(1,0,0,-1)})^{\ot 2}.$$
 Since
$\cQ_{(0)}^4$ is the trivial representation, $\cQ_{(0)}^4 \ot
\cQ_{(0)}^4$ is trivial as well, and the projector onto it is given by
$\proj{00}$.  Next $\cQ_{(0)}^4 \ot \cQ_{(1,0,0,-1)}^4   
\cong \cQ_{(1,0,0,-1)}^4$, and thus has no invariant subspace.
Similarly for $\cQ^4_{(1,0,0,-1)} \ot \cQ^4_{(0)}$ has no invariant
subspace.
$\cQ_{(1,0,0,-1)}^4$ is self-dual, and thus by Schur's Lemma,
$\cQ_{(1,0,0,-1)}^4 \ot \cQ_{(1,0,0,-1)}^4$ has a one-dimensional invariant
subspace.  To determine this subspace we observe that in the basis
$\{\ket{p}\}_{p\in\{0,1,2,3\}^2\backslash\{(0,0)\}}$, the
representation matrices of $\cQ_{(1,0,0,-1)}^4$ are real.  Also $(A
\ot I)\ket{\xi} = (I \ot A^T)\ket{\xi}$ for any matrix $A$ and for
$\ket{\xi}$ defined in \eq{xi-def}.  Together this means that
$\ket{\xi}$ is an invariant vector in
$\cQ_{(1,0,0,-1)}^4 \ot \cQ_{(1,0,0,-1)}^4$, and in fact spans its
invariant subspace.  We conclude that $\E_W \ad_W^{\ot 2} = \proj{00}
+ \proj{\xi}$. 

We now turn to the analysis of the classical
Markov chain.  Similar arguments were used in \cite{DOP07} and a
tighter analysis is forthcoming in \cite{HL08}.   First, we claim that
the Markov chain is ergodic and 
irreducible outside the vertex $0^n$.  Irreducibility follows from the
fact that every nonzero vertex is connected to $1^n$, while the
presence of self-loops implies ergodicity.  We can verify the
stationary distribution from the detailed balance condition using a
short calculation.

For the gap we will use
the comparison method of Diaconis and 
Saloff-Coste~\cite{DS96}.  Consider a Markov chain that picks two
random sites and replaces them each with random numbers from
$\{0,1,2,3\}$ subject only to avoiding the state $0^n$.  It follows
from \cite[Thm 3.2]{DS96}
that this chain has gap $\geq 2/n$, and applying the comparison method
(\cite[Thm 3.3]{DS96}) to the Markov chain described in
\lemref{markov} yields a lower bound of $1/n^2$ for its gap.  

To prove the final claim of \lemref{markov}, we want to choose $C$
sufficiently large so that 
\be |\E\gamma_t^2(p) - 4^{-n}/(1+2^{-n})| \leq 2^{-4n}
\label{eq:markov-convergence}\ee
whenever $t\geq Cn^3$.  Note that any computational basis state has
overlap $\exp(-O(n))$ with the stationary distribution.  Since the gap
is $\Omega(1/n^2)$, we can then use standard bounds on the mixing time of
Markov chains (e.g.  \cite[Lemma 2.8]{DS96}) to show that
\eq{markov-convergence} holds when $t\geq Cn^3$ for some constant $C$.

\section{Proof of \lemref{complex-approx}}
\label{app:complex-approx}

  We want to bound
$$\max_{\vec{\theta}} \l|\sum_{k=1}^d \theta_k x_k \r|
= \max_\phi \max_{\vec{\theta}} \text{Re}\l(e^{i\phi}\sum_{k=1}^d \theta_k x_k\r).$$
Here $\phi$ is maximized over $[0, 2\pi]$, and $\text{Re}(z)$ refers to
the real part of a complex number $z$.  We can now move the
maximization over $\vec{\theta}$ inside the sum to obtain \ba
\max_{\phi}\max_{\vec{\theta}} \text{Re}\l(e^{i\phi}\sum_{k=1}^d
\theta_k x_k\r) &= 
\max_\phi\sum_{k=1}^d \max_{\theta_k\in\{\pm 1\}} \text{Re}\l(e^{i\phi}
\theta_k x_k\r)\\
&= \max_\phi\sum_{k=1}^d \l|\text{Re}(e^{i\phi} x_k)\r|
\geq \E_\phi \sum_{k=1}^d \l|\text{Re}(e^{i\phi} x_k)\r|\\
&= \E_\phi \sum_{k=1}^d |x_k|\cdot |\cos \phi| = \frac{2}{\pi}
\sum_{k=1}^d |x_k|. 
\ea $\E_\phi$ indicates the expectation over $\phi$ chosen uniformly at
random from $[0,2\pi]$, and the second to last equality uses the
rotational invariance of the distribution of $\phi$.
\qed

\end{document}